\documentclass[reqno]{amsart}
\usepackage{amssymb,eucal}
\usepackage[usenames]{color}

\usepackage{amscd}

\usepackage{graphics}

\usepackage{mathrsfs}

\theoremstyle{plain}

\newtheorem{lemma}{Lemma}[section]
\newtheorem{theorem}[lemma]{Theorem}

\newtheorem{corollary}[lemma]{Corollary}

\newtheorem*{sclaim}{Claim}

\newtheorem*{stat}{\name}
\newcommand{\name}{testing}

\theoremstyle{definition}
\newtheorem{definition}[lemma]{Definition}

\theoremstyle{remark}
\newtheorem{remark}[lemma]{Remark}

\newtheorem*{remark*}{Remark}

\newcommand{\qedc}{{\qed}~{\rm Claim~{\theclaim}.}}
\newcommand{\qedsc}{{\qed}~{\rm Claim.}}

\newenvironment{scproof}
{\begin{proof}[Proof of Claim.]}
{\qedsc\renewcommand{\qed}{}\end{proof}}

\numberwithin{equation}{section}

\newcommand{\set}[1]{\{#1\}}
\newcommand{\setm}[2]{\set{#1\mid#2}}

\newcommand{\genp}[1]{\langle{#1}\rangle_+}

\newcommand{\cA}{\mathcal{A}}

\DeclareMathOperator{\length}{lh}
\DeclareMathOperator{\card}{card}
\DeclareMathOperator{\id}{id}

\DeclareMathOperator{\End}{End}

\begin{document}

\title{The finiteness problem for automaton semigroups is undecidable}

\author[P.~Gillibert]{Pierre Gillibert}
\address{Laboratoire d'Informatique Algorithmique: Fondements et Applications,
CNRS UMR 7089, Universit\'e Paris Diderot - Paris 7, Case 7014,
75205 Paris Cedex 13}
\email{pgillibert@yahoo.fr}
\email{gillibert@liafa.univ-paris-diderot.fr}


\date{\today}

\begin{abstract}
The finiteness problem for automaton groups and semigroups has been widely studied, several partial positive results are known. However we prove that, in the most general case, the problem is undecidable.

We study the case of automaton semigroups. Given a NW-deterministic Wang tile set, we construct a Mealy automaton, such that the plane admits a valid Wang tiling if and only if the Mealy automaton generates a infinite semigroup. The construction is similar to a construction by Kari for proving that the nilpotency problem for cellular automata is unsolvable.

Moreover Kari proves that the tiling of the plane is undecidable for NW-deterministic Wang tile set. It follows that the finiteness problem for automaton semigroups is undecidable.
\end{abstract}
\maketitle

\section{Introduction}

Automaton groups, where first introduced by Glu\v skov in \cite{ATA}. This family of groups is a powerful tool to build examples or counter-examples to various problems in group theory. Ale\v sin in \cite{FABPPG} constructs a new counter-example to the unbounded Burnside problem. Grigorchuk gave in \cite{OBPPG} an infinite $2$-group $G$ generated by three involutions, giving another counter-example to the unbounded Burnside problem. Grigorchuk solves the Milnor problem in \cite{DGFGG,OMPG}, proving that $G$ is of intermediate growth (its growth is neither polynomial nor exponential). Grigorchuk also proved in \cite{G:Day} that $G$ is amenable but not elementary amenable, giving the first counterexample to the Day problem \cite{Day}. Sushchansky, Gupta, and Sidki gave in \cite{Sus,GS} examples of infinite $p$-groups generated by two elements, for each prime $p>2$. Wilson in \cite{EGUEGG} answers a question by Gromov, constructing an example of group with exponential growth but without uniform exponential growth. Grigorchuk and \.Zuk proved in \cite{GZ} that the lamplighter group is an automaton group generated by a $2$-state automaton. Further study with Linnell and Schick in \cite{GLSZ} led to a counterexample to the strong Atiyah conjecture.

An automaton group is generated by the states of a finite Mealy automaton. Therefore it is natural to ask which classical group-theoretical questions are decidable.

For example, the word problem is decidable. There is an algorithm which, given an automaton group (or automaton semigroup) and two words in the generators, decides whether or not the words represent the same element. On the other hand \v Suni\'c and Ventura construct in \cite{CPAGIS} examples of automaton groups in which the conjugacy problem is not solvable.

We refer to \cite[Section~7]{ADSG}, for a list of several decision problems on automaton semigroups. The finiteness problem has been widely studied, several partial positive results are known. For example, Klimann proves in \cite{K12} that the finiteness problem is solvable among invertible-reversible Mealy automata with two states (or two letters). Mintz solved the finiteness problem for Cayley (automaton) semigroup in \cite{M}, let $S$ be a finite semigroup, the Cayley semigroup of $S$ is finite if and only if $S$ is aperiodic. There is a summary of other positive results in \cite{FPAG}.

In this paper we prove that the finiteness problem for automaton semigroups is not solvable.

The proof relies on a construction by Kari in \cite{TNPCA}. Kari constructs, given a NW-deterministic tile set $T$, a cellular automaton $C_T$, such that the plane has valid tiling in $T$ if and only if $C_T$ is not nilpotent. Kari also proves that the tiling problem for NW-deterministic tile set is unsolvable, hence the nilpotency problem for cellular automata is undecidable.

Since cellular automata are similar to Mealy automata, Kari's construction in \cite{TNPCA} can be adapted to Mealy automata. Given a NW-deterministic tile set $T$ we construct a Mealy automaton $\cA_T$ such that the plane has valid tiling in $T$ if and only if the semigroup $\genp{\cA_T}$ generated by $\cA_T$ is infinite, hence the finiteness problem for automaton semigroups is also undecidable.

The problem is still open for automaton groups. Although the methods of Lecerf in \cite{MTR}, the result of Kari and Ollinger in \cite{PIRC}, proving that periodicity is undecidable for cellular automata, suggest that the finiteness problem is also undecidable for automaton groups. The methods of \cite{M,SS} might also be useful.

\section{Basic concepts}

We denote $\mathbb{N}=\set{0,1,2,\dots}$ the set of all nonnegative integers.

Given a set $X$ and $n\in\mathbb{N}$, we denote by $X^n$ the set of all words of length $n$ over~$X$, that is the set of all sequences $u=(x_1,\dots,x_n)$ with entries in $X$, we set $\length u=n$. The only word of length 0, or equivalently the empty word, is denoted by $\varepsilon$. We also denote by $X^\omega$ the set of all infinite sequences $(x_k)_{k\in\mathbb{N}}$ with entries in $X$ and by $X^*$ the set of all finite words, that is:
\begin{equation*}
X^*=\bigcup_{n<\omega} X^n\,.
\end{equation*}
Furthermore we set:
\begin{align*}
X^{\le n}&=\setm{u\in X^*}{\length{u}\le n}=\bigcup_{k\le n} X^k\,.\\
X^{<   n}&=\setm{u\in X^*}{\length{u}<   n}=\bigcup_{k <  n} X^k\,.
\end{align*}
Given $u\in X^*$ and $v\in X^*\cup X^\omega$, we denote by $uv$ the concatenation of the words $u$ and $v$. Given $x\in X$ we denote by $x^n$ the constant sequence of length $n$ which takes the value $x$ for all indices, and by $x^\omega=(x)_{k\in\mathbb{N}}$ the infinite constant sequence.

A Mealy automaton $\cA$ is a $4$-tuple $(A,\Sigma,\delta,\sigma)$ where $A$ and $\Sigma$ are finite sets, $\delta\colon A\times\Sigma\to A$ and $\sigma\colon A\times\Sigma\to\Sigma$ are maps, called the \emph{transition} and the \emph{output} maps.

We extend the maps $\sigma\colon A^*\times\Sigma^{\le\omega}\to\Sigma^{\le\omega}$ and $\delta\colon A^{\le\omega}\times \Sigma^*\to A^{\le\omega}$ in the usual way. We also denote
$\sigma_a(u)=\sigma(a,u)$ and $\delta_u(a)=\delta(a,u)$, for all $a\in A^*$ and all $u\in\Sigma^*$. The equalities \eqref{E:aut1}-\eqref{E:aut4} are satisfied, indeed these equalities define the extensions of the maps $\delta$ and $\sigma$.
\begin{equation}\label{E:aut1}
\sigma_a(uv)=\sigma_a(u)\sigma_{\delta_u(a)}(v)\,,\quad\text{for all $a\in A^*$, $u\in\Sigma^*$, and $v\in\Sigma^*\cup\Sigma^{\omega}$.} 
\end{equation}
\begin{equation}\label{E:aut2}
\delta_u(ab)=\delta_u(a)\delta_{\sigma_a(u)}(b)\,,\quad\text{for all $u\in \Sigma^*$, $a\in A^*$, and $b\in A^*\cup A^{\omega}$.} 
\end{equation}
\begin{equation}\label{E:aut3}
\sigma_{ab}=\sigma_b\circ\sigma_a\,,\quad\text{for all $a,b\in A^*$.} 
\end{equation}
\begin{equation}\label{E:aut4}
\delta_{uv}=\delta_v\circ\delta_u\,,\quad\text{for all $u,v\in \Sigma^*$.} 
\end{equation}
Note that, given $a\in A^*$, the map $\sigma_a$ preserves the length of each word $u\in\Sigma^{\le\omega}$, moreover if $u$ is a prefix of $v$, then $\sigma_a(u)$ is a prefix of $\sigma_a(v)$. That is $\sigma_a$ is an endomorphism of the tree $\Sigma^*$.

We denote by $\genp{\cA}$ the subsemigroup of $\End\Sigma^*$ generated by $\setm{\sigma_a}{a\in A}$, equivalently $\genp{\cA}=\setm{\sigma_a}{a\in A^*\setminus\set{\varepsilon}}$.

\section{Mealy automata from NW-determinisc tile sets}

The following definition is due to Wang \cite{PTPR}.

\begin{definition}
A \emph{Wang tile} is a tuple $t=(t_N,t_S,t_E,t_W)$, where $t_N, T_S, T_E$, and $T_W$ are elements of a set of colors, viewed as a square with colored edges. A \emph{tile set} is a finite set of Wang tiles. A \emph{Wang tiling} of a subset~$P$ of~$\mathbb{Z}^2$, with a tile set~$T$, is a map $t\colon P\to T$. We say that $t$ is \emph{valid} if, given $(x,y)\in\mathbb{Z}^2$, the following equalities hold
\begin{align*}
t(x,y)_N &=t(x,y+1)_S, &\text{if $(x,y)\in P$ and $(x,y+1)\in P$}.\\
t(x,y)_E &=t(x+1,y)_W, &\text{if $(x,y)\in P$ and $(x+1,y)\in P$}.
\end{align*}
\end{definition}

A simple compactness argument gives the following classical result.

\begin{theorem}\label{T:compactness}
Let $T$ be a tile set. The set $\mathbb{Z}^2$ has a valid Wang tiling if and only if each finite subset of $\mathbb{Z}^2$ has a valid Wang tiling.
\end{theorem}

\begin{remark}
In particular, if $\mathbb{Z}^2$ has no valid Wang tiling, then there is the least integer $n\in\mathbb{N}$ such that $\set{0,1,\dots,n}^2$ has no valid Wang tiling.
\end{remark}

The existence of valid Wang tiling is hard to determine, as the following result of R. Berger in \cite{UDP} illustrates.

\begin{theorem}[Berger]\label{T:UTS}
It is undecidable whether or not a finite tile set has a valid Wang tiling for $\mathbb{Z}^2$.
\end{theorem}

The following notion was introduced by Kari in \cite{TNPCA}.

\begin{definition}
A tile set $T$ is \emph{NW-deterministic} if each tile is determined by the north and west colors. That is $t_N=s_N$ and $t_W=s_W$ imply that $t=s$, for all $s,t\in T$.
\end{definition}

Theorem~\ref{T:UTS} is generalized by Kari in \cite{TNPCA}.

\begin{theorem}[Kari]\label{T:UTS2}
It is undecidable whether or not a finite NW-deterministic tile set has a valid Wang tiling for $\mathbb{Z}^2$.
\end{theorem}

The main goal was to generalize a result of Culik, Pachl, and Yu in \cite{OLSCA} to dimension one. Kari proves the following theorem in \cite{TNPCA}.

\begin{theorem}[Kari]
It is undecidable whether or not a one-dimensional cellular automaton is nilpotent.
\end{theorem}

The argument can be adapted to automaton semigroups, however we need to be careful about a side effect. A cellular automaton acts on words indexed by $\mathbb{Z}$, while each element of an automaton semigroup acts on words indexed by $\mathbb{N}$. We first define a Mealy automaton from a tile set (Kari uses a similar construction to obtain a cellular automaton).

\begin{definition}\label{D:MAformDTS}
Let $T$ be a finite NW-deterministic tile set, and let $\bot$ denote a special symbol which is not in T. The \emph{Mealy automaton of $T$} is the tuple $\cA_T=(A,\Sigma,\delta,\sigma)$, where $A=\Sigma=T\sqcup\set{\bot}$, and the maps $\delta$ and $\sigma$ are defined by
\begin{align*}
\delta\colon A\times\Sigma &\to A\\
(x,y)&\mapsto y
\end{align*}
The new state does not depend on the old one, the automaton only remembers the previous letter.
\begin{align*}
\sigma\colon A\times\Sigma&\to \Sigma\\
(\bot,s)&\mapsto \bot\\
(t,\bot)&\mapsto \bot\\
(\bot,\bot)&\mapsto \bot\\
(s,t)&\mapsto r &&\text{if $r_N=t_S$ and $r_W=s_E$.}\\
(s,t)&\mapsto \bot &&\text{otherwise.}
\end{align*}
That is, given $s,t,r\in T$, if the Wang tiling on the left hand side of Figure~\ref{F:AWT} is valid, then $\sigma(s,t)=r$, in all other cases $\sigma(s,t)=\bot$.
\end{definition}

\begin{remark}
The Mealy automaton of a finite NW-deterministic tile set $T$ should be understood in the following way. A word $w$ in $A$, can be seen as a word written on tiles along the diagonal $D$, the Mealy automaton transforms this word to the word written on the tiles along the diagonal right below the diagonal $D$. If it is impossible to put a tile at some place, then the ``mistake'' tile $\bot$ is placed instead.

The Mealy automaton $\cA_T$ is a reset automaton. Silva and Steinberg have studied groups and semigroups generated by invertible reset automata. In particular such group is infinite if and only if any generator is of infinite order (cf. \cite[Theorem~3.2]{SS}). A generalization of this paper to automaton groups would required to prove that this problem is also undecidable.
\end{remark}

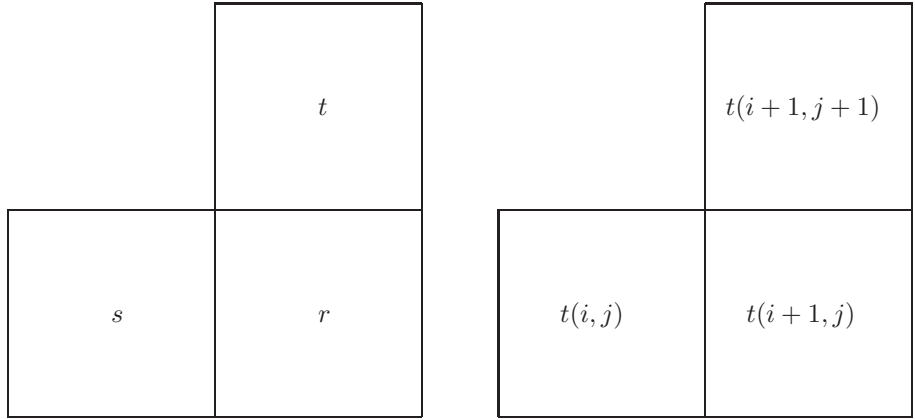
\begin{figure}
\setlength{\unitlength}{0.55mm}
\begin{picture}(110,110)(-5,-5)
\put(0,0){\line(1,0){100}}
\put(0,50){\line(1,0){100}}
\put(50,100){\line(1,0){50}}
\put(0,0){\line(0,1){50}}
\put(50,0){\line(0,1){100}}
\put(100,0){\line(0,1){100}}
\put(25,23){$s$}
\put(75,23){$r$}
\put(75,73){$t$}
\end{picture}
\quad
\begin{picture}(110,110)(-5,-5)
\put(0,0){\line(1,0){100}}
\put(0,50){\line(1,0){100}}
\put(50,100){\line(1,0){50}}
\put(0,0){\line(0,1){50}}
\put(50,0){\line(0,1){100}}
\put(100,0){\line(0,1){100}}
\put(15,23){$t(i,j)$}
\put(60,23){$t(i+1,j)$}
\put(55,73){$t(i+1,j+1)$}
\end{picture}
\caption{Wang tilings.}\label{F:AWT}
\end{figure}

\begin{remark}\label{R:SigmaAppliedToWord}
Note that $\delta_x(a)=\delta(a,x)=x$ for all $a\in A$ and $x\in\Sigma$. It follows that:
\begin{equation}\label{E:SigmaAppliedToWord}
\sigma_a(u)=\sigma_a(u_0) (\sigma_{u_k}(u_{k+1}))_{k\in\mathbb{N}},\quad\text{for all $u=(u_k)_{k\in\mathbb{N}}\in\Sigma^\omega$, and all $a\in A$}.
\end{equation}
\end{remark}

\begin{lemma}\label{L:reducepart1}
Let $T$ be a finite NW-deterministic tile set. Let $t\colon \mathbb{Z}^2\to T$ be a valid Wang tiling. Consider the word $w_n=(t(k+n,k))_{k\in\mathbb{N}}$ for each $n\in\mathbb{N}$. The equality $\sigma_\bot^m(w_n)=\bot^m w_{m+n}$ holds for all $n,m\in\mathbb{N}$. In particular all the maps $\sigma_\bot^m$ are different.
\end{lemma}

\begin{proof}
We use the notations of Definition~\ref{D:MAformDTS}. As the Wang tiling on the right hand side of Figure~\ref{F:AWT} is valid, it follows that
\begin{equation}\label{E:sigmavalidtiling}
\sigma_{t(i,j)}(t(i+1,j+1))=t(i+1,j),\quad\text{for all $i,j\in\mathbb{N}$}.
\end{equation}
Given $n\in\mathbb{N}$, the following equalities hold.
\begin{align*}
\sigma_\bot(w_n)&=\sigma_\bot(t(n,0))(\sigma_{t(n+k,k)}( t(n+k+1,k+1)))_{k\in\mathbb{N}}\,, &&\text{by \eqref{E:SigmaAppliedToWord}.}\\
&=\bot(t(n+k+1,k))_{k\in\mathbb{N}}\,, &&\text{by \eqref{E:sigmavalidtiling}.}\\
&=\bot w_{n+1}
\end{align*}
The result follows by induction.
\end{proof}

\begin{lemma}\label{L:reducepart2}
Let $T$ be a finite NW-deterministic tile set. If $\mathbb{Z}^2$ has no valid Wang tiling then $\genp{\cA_T}$ is finite.
\end{lemma}

\begin{proof}
We use the notations of Definition~\ref{D:MAformDTS}. By Theorem~\ref{T:compactness} there is $n\in\mathbb{N}$ such that the set $\set{0,1,\dots,n}^2$ has no valid Wang tiling for $T$. 

\begin{sclaim}
Let $u\in A^{2n}$. The following equality holds.
\[
\sigma_u(pq)=\sigma_u(p)\bot^{\omega}\,,\quad\text{for all $p\in\Sigma^n$ and all $q\in\Sigma^\omega$.}
\]
\end{sclaim}
\begin{scproof}
We can write $u=u_1\dots u_{2n}$. Set $\tau_0=\id$, and set:
\[
 \tau_k=\sigma_{u_1u_2\dots u_k}=\sigma_{u_k}\circ\sigma_{u_{k-1}}\circ\dots\circ\sigma_{u_{1}}\,,\quad\text{for each $1\le k \le 2n$.}
\]
Notice that
\begin{equation}\label{E:composetau}
\sigma_{u_{k+1}}\circ\tau_k=\tau_{k+1}\,,\quad\text{for all $0\le k \le 2n-1$.}
\end{equation}

Let $p\in\Sigma^n$, let $q\in\Sigma^\omega$. Denote by $f(i,j)$ the $j^{\text{th}}$ letter of $\tau_{i}(pq)$, for $(i,j)\in\mathbb{N}^2$ such that $i\le 2n$. That is:
\begin{equation}\label{E:taupq}
 \tau_{i}(pq)=f(i,j)_{j\in\mathbb{N}}\,,\quad\text{for all $0\le i\le 2n$.}
\end{equation}

Given $0\le i< 2n$, the following equalities hold:
\begin{align*}
f(i+1,j)_{j\in\mathbb{N}} &= \tau_{i+1}(pq)\,,&&\text{by \eqref{E:taupq}}\\
&=\sigma_{u_{i+1}}(\tau_i(pq))\,,&&\text{by \eqref{E:composetau}}\\
&=\sigma_{u_{i+1}}(f(i,j)_{j\in\mathbb{N}})\,,&&\text{by \eqref{E:taupq}}\\
&=\sigma_{u_{i+1}}(f(i,0)) (\sigma_{f(i,j)} (f(i,j+1)))_{j\in\mathbb{N}} \,,&&\text{by \eqref{E:SigmaAppliedToWord} in Remark~\ref{R:SigmaAppliedToWord}}
\end{align*}
Therefore the following statement holds
\begin{equation}\label{E:vwtlocal}
\sigma_{f(i,j)} (f(i,j+1))= f(i+1,j+1)\,,\quad\text{for all $(i,j)\in\mathbb{N}^2$ with $0\le i< 2n$.}
\end{equation}

Assume that $f(2n,n+k)\not=\bot$ for some $k\in\mathbb{N}$. Applying inductively \eqref{E:vwtlocal}, with Definition~\ref{D:MAformDTS} we obtain that $f(i+j,i+k)\not=\bot$ for all $0\le i,j \le n$, and the $n\times n$ Wang tiling on Figure~\ref{F:AWT2} is valid.

Therefore $\set{0,1,\dots,n}^2$ has a valid Wang tiling; a contradiction.
\end{scproof}

Let $u\in A^*$ be a word of length at least $2n$, let $v\in A^{2n}$ and $w\in A$ be such that $u=vw$. Let $p\in\Sigma^n$, let $q\in \Sigma^\omega$. We have
\[
\sigma_u(pq)=\sigma_{vw}(pq)=\sigma_w(\sigma_v(pq))=\sigma_w(\sigma_v(p)\bot^\omega)=\sigma_{vw}(p)\bot^\omega=\sigma_u(p)\bot^\omega\,.
\]
Therefore $\setm{\sigma_u}{u\in A^*\text{ and }\length{u}\ge 2n}$ is of cardinality at most $\card (\Sigma^n)^{(\Sigma^n)}$.

However $\genp{\cA_T}= \setm{\sigma_u}{u\in A^{<2n}}\cup \setm{\sigma_u}{u\in A^*\text{ and }\length{u}\ge 2n}$, therefore the following inequality holds 
\[
\card\genp{\cA_T}\le 1+\card A+\card A^2+\dots+\card A^{2n-1} + \card (\Sigma^n)^{(\Sigma^n)}\,.
\]
Hence $\genp{\cA_T}$ is finite.
\end{proof}

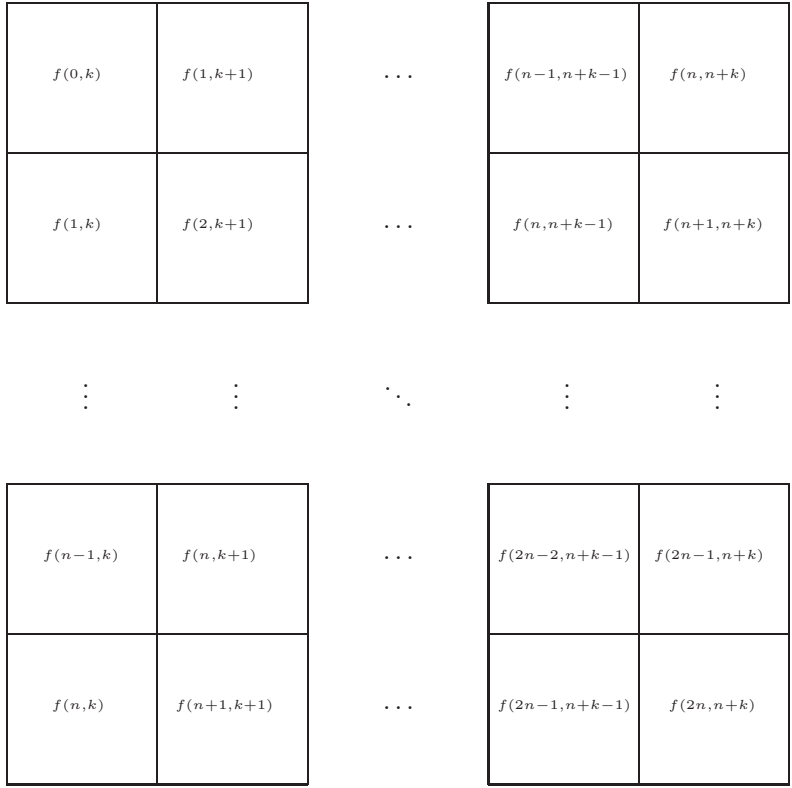
\begin{figure}
\setlength{\unitlength}{0.4mm}
\begin{picture}(270,270)(-5,-5)
\put(0,0){\line(0,1){100}}
\put(50,0){\line(0,1){100}}
\put(100,0){\line(0,1){100}}
\put(160,0){\line(0,1){100}}
\put(210,0){\line(0,1){100}}
\put(260,0){\line(0,1){100}}

\put(0,160){\line(0,1){100}}
\put(50,160){\line(0,1){100}}
\put(100,160){\line(0,1){100}}
\put(160,160){\line(0,1){100}}
\put(210,160){\line(0,1){100}}
\put(260,160){\line(0,1){100}}

\put(0,0){\line(1,0){100}}
\put(0,50){\line(1,0){100}}
\put(0,100){\line(1,0){100}}
\put(0,160){\line(1,0){100}}
\put(0,210){\line(1,0){100}}
\put(0,260){\line(1,0){100}}

\put(160,0){\line(1,0){100}}
\put(160,50){\line(1,0){100}}
\put(160,100){\line(1,0){100}}
\put(160,160){\line(1,0){100}}
\put(160,210){\line(1,0){100}}
\put(160,260){\line(1,0){100}}

\put(15,235){$\scriptscriptstyle{f(0,k)}$}
\put(58,235){$\scriptscriptstyle{f(1,k+1)}$}
\put(125,235){$\dots$}
\put(165,235){$\scriptscriptstyle{f(n-1,n+k-1)}$}
\put(220,235){$\scriptscriptstyle{f(n,n+k)}$}

\put(15,185){$\scriptscriptstyle{f(1,k)}$}
\put(58,185){$\scriptscriptstyle{f(2,k+1)}$}
\put(125,185){$\dots$}
\put(168,185){$\scriptscriptstyle{f(n,n+k-1)}$}
\put(218,185){$\scriptscriptstyle{f(n+1,n+k)}$}

\put(25,125){$\vdots$}
\put(75,125){$\vdots$}
\put(125,125){$\ddots$}
\put(185,125){$\vdots$}
\put(235,125){$\vdots$}

\put(12,75){$\scriptscriptstyle{f(n-1,k)}$}
\put(58,75){$\scriptscriptstyle{f(n,k+1)}$}
\put(125,75){$\dots$}
\put(163,75){$\scriptscriptstyle{f(2n-2,n+k-1)}$}
\put(215,75){$\scriptscriptstyle{f(2n-1,n+k)}$}

\put(15,25){$\scriptscriptstyle{f(n,k)}$}
\put(56,25){$\scriptscriptstyle{f(n+1,k+1)}$}
\put(125,25){$\dots$}
\put(163,25){$\scriptscriptstyle{f(2n-1,n+k-1)}$}
\put(220,25){$\scriptscriptstyle{f(2n,n+k)}$}

\end{picture}
\caption{A Wang tiling defined by an element of $\genp{\cA_T}$.}\label{F:AWT2}
\end{figure}

From Lemma \ref{L:reducepart1} and Lemma~\ref{L:reducepart2} we see that  the existence of a valid Wang tiling of $\mathbb{Z}^2$ is equivalent to the infiniteness of an explicit automaton semigroup. Therefore, from Theorem~\ref{T:UTS2} we deduce the following result.

\begin{theorem}
It is undecidable whether or not a given automaton semigroup is finite.
\end{theorem}

From the proof of Lemma~\ref{L:reducepart2}, we see the following corollary.

\begin{corollary}
It is undecidable whether or not, given an automaton semigroup $A$ and $f,g\in A$, there exists $n$ such that $f^n=g$.
\end{corollary}

\begin{proof}
Given a finite NW-deterministic tile set $T$, we consider the Mealy automaton $\cA_T=(A,\Sigma,\delta,\sigma)$ defined in Definition~\ref{D:MAformDTS}. We add an additional state $c$ to $A$ and extend $\sigma$ and $\delta$ by:
\begin{align*}
\sigma(c,x)=\bot && \text{for all $x\in\Sigma$.}\\
\delta(c,x)=c    && \text{for all $x\in\Sigma$.}
\end{align*}
We obtain a new Mealy automaton. The corresponding automaton semigroup contains a new element $\sigma_c$. Notice that $\sigma_c(w)=\bot^\omega$ for each infinite word $w\in\Sigma^\omega$.

From the Lemma~\ref{L:reducepart1} and Lemma~\ref{L:reducepart2} we see that the following statement are equivalent:
\begin{enumerate}
\item The exists a positive integer $n$ such that $\sigma_\bot^n=\sigma_c$.
\item There is no valid tiling of $\mathbb{Z}^2$ with $T$.
\end{enumerate}
The contrapositive of $(1)\Longrightarrow (2)$ is a direct consequence of Lemma~\ref{L:reducepart1}. Notice that $\sigma_\bot^n(w)$ always start with $n$ times the symbol $\bot$, for each word $w\in\Sigma^\omega$. It follows from the proof of Lemma~\ref{L:reducepart2} that $(2)\Longrightarrow (1)$.

However the tiling problem is undecidable for NW-deterministic tile set (cf. Theorem~\ref{T:UTS2}), therefore $(1)$ is undecidable.
\end{proof}


\begin{thebibliography}{99}


\bibitem{FPAG}
A. Akhavi, I. Klimann, S. Lombardy, J. Mairesse, and M. Picantin.
\emph{On the Finiteness Problem for Automaton (Semi)groups},
International Journal of Algebra and Computation \textbf{22} (2012), No.~6, 26 pp.

\bibitem{FABPPG}
S.~V. Ale\v sin,
\emph{Finite automata and the Burnside problem for periodic groups.} (Russian),
Mat. Zametki \textbf{11} (1972), 319--328.


\bibitem{UDP}
R. Berger,
\emph{The undecidability of the Domino problem},
Mem. Amer. Math. Soc. \textbf{66} (1966), 72 pp.


\bibitem{OLSCA}
K. Culik, J. Pachl, and S. Yu,
\emph{On the limit sets of cellular automata},
SIAM J. Comput. \textbf{18} (1989) 831--842.

\bibitem{Day}
M.~M. Day,
\emph{Amenable semigroups},
Illinois J. Math. \textbf{1} (1957), 509--544.


\bibitem{ATA}
V.~M. Glu\v skov, \emph{Abstract theory of automata.} (Russian),
Uspehi Mat. Nauk \textbf{16} (1961) no. 5 (101), 3--62. 


\bibitem{OBPPG}
R.~I. Grigorchuk,
\emph{On Burnside's problem on periodic groups.} (Russian),
Funktsional. Anal. i Prilozhen. \textbf{14} (1980), no. 1, 53--54.

\bibitem{OMPG}
R.~I. Grigorchuk,
\emph{On the Milnor problem of group growth.} (Russian),
Dokl. Akad. Nauk SSSR \textbf{271} (1983), no.~1, 30--33.

\bibitem{DGFGG}
R.~I. Grigorchuk,
\emph{Degrees of growth of finitely generated groups and the theory of invariant means.} (Russian),
Izv. Akad. Nauk SSSR Ser. Mat. \textbf{48} (1984), no.~5, 939--985.

\bibitem{G:Day}
R.~I. Grigorchuk,
\emph{On a problem of M. Day on nonelementary amenable groups in the class of finitely presented groups.} (Russian),
Mat. Zametki \textbf{60} (1996), no.~5, 774--775.

\bibitem{GLSZ}
R.~I. Grigorchuk, P. Linnell, T. Schick, and A. \.Zuk,
\emph{On a question of Atiyah},
C. R. Acad. Sci. Paris S\'er. I Math. \textbf{331} (2000), no.~9, 663--668.

\bibitem{ADSG}
R.~I. Grigorchuk, V.~V. Nekrashevich, and V.~I. Sushchanski\u{\i},
\emph{Automata, Dynamical Systems, and Groups},
Proc. Steklov Inst. Math. \textbf{231} (2000), 134--214.

\bibitem{GZ}
R.~I. Grigorchuk and A. \.Zuk,
\emph{The lamplighter group as a group generated by a 2-state automaton, and its spectrum},
Geom. Dedicata \textbf{87} (2001), no.~1-3, 209--244.

\bibitem{GS}
N. Gupta and S. Sidki,
\emph{On the Burnside problem for periodic groups},
Mat. Zametki \textbf{182} (1983), no.~3, 385--388.

\bibitem{TNPCA}
J. Kari,
\emph{The nilpotency problem of one-dimensional cellular automata},
SIAM J. Comput. \textbf{21} (1992) 571--586.


\bibitem{PIRC}
J. Kari and N. Ollinger,
\emph{Periodicity and immortality in reversible computing},
Mathematical foundations of computer science 2008, 
Lecture Notes in Comput. Sci. \textbf{5162} (2008), 419--430, Springer, Berlin. 

\bibitem{K12}
I. Klimann,
\emph{The finiteness of a group generated by a 2-letter invertible-reversible Mealy automaton is decidable},
30th International Symposium on Theoretical Aspects of Computer Science, Kiel, Germany,
Leibniz International Proceedings in Informatics \textbf{20} (2013), 502--513, Natacha Portier and Thomas Wilke (Eds.).

 
\bibitem{MTR}
Y. Lecerf,
\emph{Logique Math\'ematique. Machines de Turing r\'eversibles. R\'ecursive insolubilit\'e en $n\in\mathbb{N}$ de l'\'equation $u=\theta^nu$, o\`u $\theta$ est un ``isomorphisme de codes''},
Comptes Rendus Hebdomadaires des S\'eances de L'acad\'emie des Sciences, \textbf{257} (1963) 2597--2700.


\bibitem{M}
A. Mintz,
\emph{On the Cayley semigroup of a finite aperiodic semigroup},
Internat. J. Algebra Comput. \textbf{19} (2009), no.~6, 723--746. 


\bibitem{SS}
P.~V. Silva and B. Steinberg,
\emph{On a class of automata groups generalizing lamplighter groups},
Internat. J. Algebra Comput. \textbf{15} (2005), no. 5-6, 1213--1234.

\bibitem{CPAGIS}
Z. \v Suni\'c and E. Ventura,
\emph{The conjugacy problem in automaton groups is not solvable},
J. Algebra \textbf{364} (2012), 148--154.

\bibitem{Sus}
V.~I. Sushchansky,
\emph{Periodic p-groups of permutations and the unrestricted Burnside problem.} (Russian),
Dokl. Akad. Nauk SSSR \textbf{247} (1979), no.~3, 557--561.



\bibitem{PTPR}
H. Wang,
\emph{Proving theorems by pattern recognition--II},
Bell System Tech. \textbf{40} (1961) 1--42.

\bibitem{EGUEGG}
J.~S. Wilson,
\emph{On exponential growth and uniformly exponential growth for groups},
Invent. Math. \textbf{155} (2004), no.~2, 287--303.



\end{thebibliography}
\end{document}